\documentclass[11pt]{article}
\usepackage{hyperref}

\usepackage{bm,color,fullpage}
\usepackage{amsmath,amssymb,amsthm}
\usepackage{graphicx}
\usepackage{algorithm}
\usepackage{algorithmicx}
\usepackage{algpseudocode}
\usepackage{subfig}

\newcommand{\kh}[1] {\left(#1\right)}
\newcommand{\zkh}[1] {\left[#1\right]}
\newcommand{\dkh}[1] {\left\{#1\right\}}
\newcommand{\innerprod}[1] {\ensuremath{\left\langle #1\right\rangle}}

\newcommand{\abs}[1]{\ensuremath{\left|#1\right|}}
\newcommand{\norm}[1]{\ensuremath{\left\|#1\right\|}}

\newcommand{\dist}[1]{\ensuremath{\operatorname{dist}\kh{#1}}}
\def\erf{\ensuremath{\operatorname{erf}}}

\def\x{\ensuremath{\bm{x}}}
\def\z{\ensuremath{\bm{z}}}
\def\a{\ensuremath{\bm{a}}}
\def\A{\ensuremath{\bm{A}}}
\def\b{\ensuremath{\bm{b}}}

\def\E{\ensuremath{{\operatorname{E}  }}}
\def\R{\ensuremath{\mathbb{R}}}
\def\C{\ensuremath{\mathbb{C}}}
\def\E{\ensuremath{\text{E}}}
\def\h{\ensuremath{\bm{h}}}

\newtheorem{thm}{Theorem}
\newtheorem{property}{Property}
\newtheorem{lem}{Lemma}

\begin{document}

\title{Phase Retrieval via Smooth Amplitude Flow}

\author{Qi Luo{\thanks{
		College of Science, National University of Defense Technology,
		Changsha, Hunan, 410073, P.R.China. Email: \texttt{luoqi10@nudt.edu.cn}}}
\and Hongxia Wang{\thanks{
		College of Science, National University of Defense Technology,
		Changsha, Hunan, 410073, P.R.China. Email: \texttt{wanghongxia@nudt.edu.cn}}
}
}

\maketitle

\begin{abstract}

Phase retrieval (PR) is an inverse problem about recovering a signal from phaseless linear measurements.
This problem can be effectively solved by minimizing a nonconvex amplitude-based loss function. However, this loss function is non-smooth.
To address the non-smoothness, a series of methods have been proposed by adding truncating, reweighting and smoothing operations to adjust the gradient or the loss function and achieved better performance. But these operations bring about extra rules and parameters that need to be carefully designed. Unlike previous works, we present a smooth amplitude flow method (SAF) which minimizes a novel loss function, without additionally modifying the gradient or the loss function during gradient descending. Such a new heuristic can be regarded as a smooth version of the original non-smooth amplitude-based loss function.
We prove that SAF can converge geometrically to a global optimal point via the gradient algorithm with an elaborate initialization stage with a high probability.
Substantial numerical tests empirically illustrate that the proposed heuristic is significantly superior to the original amplitude-based loss function and SAF also outperforms other state-of-the-art methods in terms of the recovery rate and the converging speed. Specially, it is numerically shown that
SAF can stably recover the original signal when number of measurements is smaller than the information-theoretic limit for both the real and the complex Gaussian models.
\end{abstract}

\section{Introduction}
\label{sec:intro}
In various science and engineering fields, one often encounters the problem of reconstructing a signal from phaseless measurements, known as the phase retrieval (PR) problem. Specific applications of PR include X-ray crystallography \cite{miao1999extending}, molecular imaging \cite{shechtman2015phase}, biological imaging \cite{stefik1978inferring} and astronomy \cite{fienup1987phase}.

Mathematically, PR is to solve a system of quadratic equations of the form:
\begin{equation}\label{eq:pr_problem}
  b_i=\abs{\innerprod{\a_i,\x  }},\quad i=1,\cdots,m,
\end{equation}
where $\x\in\mathbb{R}^n/\mathbb C^n$ is the unknown signal to be found, the measurements $\b := [b_i]_{1\leq i\leq m}\in \mathbb{R}^m$, and $\a_i\in\R^n/\C^n$ denotes the measuring vector, forming the $m\times n$ measuring matrix $\bm A := [\a_i]_{1\leq i\leq m}$. Because $e^{j\theta}\x$ with $j:=\sqrt{-1}$ also satisfies \eqref{eq:pr_problem} for all $\theta\in\R $, the uniqueness of the phase retrieval problem is defined up to a global phase.

$\A$ usually corresponds to the discrete Fourier transform in optics. To consider more general cases,
recent works focus more on the generic measurements. The most widely studied generic measuring model is the Gaussian model, i.e.
$\a_i$ sampled from independently and identically distributed (i.i.d.) $\mathcal{N}(0,\bm I_n)$ for the real Gaussian model, or $\mathcal{CN}(0,\bm I_n) = \mathcal{N}(0,\bm I_n/2)+ j\mathcal{N}(0,\bm I_n/2)$ for the complex Gaussian model.
Under this setting, it has been proved that $m$ should be at least $4n-4$ in complex case or $2n-1$ in real case to ensure uniqueness of solution $\x$. In this sense, $m=2n-1$ and $m=4n-4$ can be regarded as the information-theoretical limits in the real case and the complex case respectively for a PR problem to be uniquely solvable.


Although PR has a simple form and wide applications across many fields, solving it meets tremendous difficulties both theoretically and numerically since it has been proved to be NP-hard in general situations \cite{ben2001lectures_s}.

\subsection{Prior art}
The mainstream classical methods to solve PR are error-reduction algorithms including the Gerchberg-Saxton, hybrid input and output methods, based on constantly alternated projections. However, fundamental mathematical questions about the convergence of these methods still remained unsolved. Recently, a convex formulation of PR was found in \cite{phaselift2015} relying on the so-called matrix-lifting technique, and
several methods abbreviated as PhaseLift \cite{phaselift2015,candes2014plift}, PhaseCut \cite{waldspurger2015phase} and CoRK \cite{huang2016phase} were proposed based on this. Many solid guarantees about perfect recovery and convergence have been established for this convex approach, but its large computational complexity makes it unpractical when the signal dimension is large.
Another convex formulation was proposed in \cite{goldstein2018phasemax} via PhaseMax, which solved a linear program in the natural parameter space. However, PhaseMax is markedly uncompetitive with other state-of-the-art methods in terms of the empirical recovery rate.

More attention was paid to non-convex formulations directly instead of convex relaxation in recent years. Relevant works include Alternating Minimization (AltMin) \cite{netrapalli2013phase}, Wirtinger Flow (WF) \cite{candes2015phase}, Amplitude Flow method (AF) \cite{zhang2016reshaped} and their variants \cite{wang2017taf,raf2018,twf2016outlier}. Specifically, WF is a gradient descent method based on minimizing the following intensity-based loss function:
\begin{equation}\label{eq:wfloss}
  \ell_{\text{WF}}(\z):= \frac{1}{2m}\sum_{i=1}^{m} (|\innerprod{\a_i,\z}|^2-b_i^2)^2
\end{equation}
WF with spectral initialization method can recover the original perfectly from $\mathcal{O}(n\log n)$ measurements. Its two variants truncated WF (TWF) and reweighted WF reduce this number to $\mathcal{O}(n)$.
Zhang found that better performance can be obtained by minimizing the following amplitude-based loss function:
\begin{equation}\label{eq:afloss}
  \ell_{\text{AF}}(\z):= \frac{1}{2m}\sum_{i=1}^{m} (|\innerprod{\a_i,\z}|-b_i)^2,
\end{equation}
and this method is known as reshaped WF (RWF). RWF is shown to achieve perfect recovery from $\mathcal{O}(n)$ measurements,
which is better than the original WF. RWF is also known as AF since it minimizes the amplitude-based loss function. To further improve the performance of AF, Wang proposed two variants, truncated AF (TAF) and reweighted AF (RAF) which respectively adopt truncating and reweighting operation during gradient researching.
In terms of convergence speed and rate of recovery, TAF and RAF exhibit a superior performance over the state-of-the-art methods. However, TAF requires a carefully selected parameter for truncation procedure up the gradient function. The reweighting procedure in RAF also requires designing several parameters to obtain a desired performance.

The non-smooth absolute value term in \eqref{eq:afloss} can deteriorate the numerical performance of AF. To tackle this, Pinilla proposed a smoothing conjugate gradient method (PR-SCG) \cite{pinilla2018phase} that adjusts the loss function slightly to make the loss function smooth.
PR-SCG is a direct application of smoothing projected gradient method \cite{zhang2009smoothing} in the PR problem.
The critical technique of PR-SCG is replacing $\abs{x}$ with a smooth function $\sqrt{\abs{x}^2+\epsilon}$.
Then the loss function of PR-SCG becomes:
\begin{equation}\label{eq:scgloss}
  \ell_{\text{PR-SCG}}(\z):= \frac{1}{2m}\sum_{i=1}^{m}\kh{ (\abs{\innerprod{\a_i,\z}}^2
  +\epsilon)^{\frac12} -b_i} ^2.
\end{equation}
Obviously, the optimization problem in PR-SCG is not equivalent to the original problem. A well-selected $\epsilon$ and a diminishing rule for $\epsilon$ is needed to ensure that the adjusted loss function \eqref{eq:scgloss} converge to the original AF loss function \eqref{eq:afloss}.

With enough measurements, the state-of-the-art methods empirically achieve perfect recovery using $\mathcal{O}(n)$ measurements under random Gaussian settings. However, all these methods require at least the information-limit number of measurements to ensure this. How to further improve the rate of recovery under the information limit remains to be exploited.

\subsection{This work}
We construct a novel loss function, which is a natural smooth version of the original amplitude-base loss function. And SAF is proposed based on minimizing such novel heuristic by the gradient descent method with a delicate initialization.
This method is simple to implement as it does not need extra operations upon the gradient or loss function as many other state-of-the-art methods do. Theoretical analysis shows that SAF will converge to the global optimum geometrically given $m=\mathcal{O}(n)$ measurements. Numerical simulations show that our SAF approach performs better than the original AF and other state-of-the-art methods in respect with the sampling complexity and time cost.

The remainder of ths paper is organized as follows. In Section \ref{sec:method}, we propose the SAF algorithm. Section \ref{sec:theo} gives the theoretical analysis of the proposed method. In section \ref{sec:numr}, various experiments are implemented to compare SAF with other gradient descent solvers.

As regards notation used in this paper, the bold capital lowercase letters, e.g. $\x,\z$ denote represent vectors. The bold capital uppercase letters such as $\A$ represent matrices. $\x'$ denotes the conjugate transpose of $\x$. $\innerprod{\x,\bm y}$ denotes the inner product of vector $\x,\bm y$ calculated by $\innerprod{\x,\bm y} = \x'\bm y.$ \norm{\x} is the Euclidean norm. The cardinality of the set $\mathcal{I}$ is denoted by $\abs{\mathcal{I}}$.
The distance between two vectors up to a global phase is defined as
$
  \dist{\x_1,\x_2}:= \min_{\theta\in \mathbb{R}} \norm{e^{j\theta}\x_1 -\x_2 }.
$

\section{Smooth amplitude flow method}
\label{sec:method}
The intuition and the principles of SAF will be presented in detail in thsi section. For concreteness, only the real Gaussian model is analyzed. However, with the aid of Wirtinger derivative SAF can be easily applied to the complex model directly.
\subsection{The smooth amplitude-based loss function}
Similar to PR-SCG, we bring in the smooth function
 \begin{equation}\label{eq:smoothfunc}
    g_{k,\epsilon}(x)=\sqrt[k]{\abs{x}^k+\epsilon^k},
 \end{equation}
where $k\geq2,\epsilon>0$. Since only replacing $\abs{\a_i'\z}$ with $g_{k,\epsilon}(\a_i'\z)$ in the original loss function \eqref{eq:afloss} will lead to a loss function with different global optimum, we consider replacing the $b_i$ in \eqref{eq:afloss} with $g_{k,\epsilon}(b_i)$ symmetrically. In addition, we set $\epsilon$ to be the proportional to $b_i$. As a result, we obtain the following globally smooth loss function:
\begin{equation}\label{eq:safloss}
  \ell_{\text{SAF}}(\z):= \frac{1}{2m}\sum_{i=1}^{m}\ell_i(\z),\quad
  \ell_i(\z)= \kh{g_{k, \gamma b_i}(\a_i'\z) - g_{k,\gamma b_i}(b_i) } ^2,
\end{equation}
where $k\geq 2,\gamma>0$ are preselected parameters. The utilization of the smooth function \eqref{eq:smoothfunc} makes the original amplitude-based loss function \eqref{eq:afloss} smooth. Therefore, our method is called the smooth amplitude flow method, abbreviated as SAF. When $\gamma=0$, $\ell_{\text{SAF}}(\z)$ degenerates to the original non-smooth loss function \eqref{eq:afloss}. Obviously, the original AF loss function and the SAF loss function \eqref{eq:safloss} have the same global minimizers that satisfy
$\abs{\innerprod{\a_i,z}} = b_i$, for $i=1,\cdots,m$. The gradient of the loss function \eqref{eq:safloss} is
\begin{equation}\label{eq:safgrad}
  \nabla \ell_{\text{SAF}}(\z) = \frac{1}{m}\sum_{i=1}^{m} \kh{g_{k, \gamma b_i}(\a_i'\z) - g_{k,\gamma b_i}(b_i) } (|\a_i'\z|^k+\gamma^kb_i^k)^{\frac 1k-1}\abs{\a_i'\z}^{k-2}\a_i\a_i'\z.
\end{equation}
It is hard to analytically determine the best setting of parameters $k$ and $\gamma$. Therefore we simply take $k=4$ and $\gamma=1$ based on extensive numerical experiments.

An earlier work proposed a similar loss function method in view of adding perturbation \cite{gao2019solving}. However, our method is proposed from the view of utilizing a smooth alternative of $\abs{x}$ and is more generalized.
Besides \eqref{eq:smoothfunc}, SAF can also adopt other smooth alternative of $\abs{x}$ such as the well-known log-cosh function \cite{choo2018symmetries}.

SAF can be seen as a natural improved version of the original AF method. According to \cite{zhang2016reshaped},
amplitude-based loss function \eqref{eq:afloss} can be regarded as a direct application of traditional least squares method to phase retrieval, and shows to be better than WF \eqref{eq:wfloss}. However, jumps of gradient $\nabla\ell_i(\z)$ exist in the vicinity of the line $\dkh{\z:\a_i'\z=0}$, which may unstabilize the gradient descent algorithm, especially when the measurement number is around the information-theoretical limit. To tackle this, AF's variants TAF and RAF modifies the gradient function, and have achieved better empirical performance. The idea of this paper is similar, but we demonstrate that directly adjusting the loss function with a smooth function \eqref{eq:smoothfunc} can also be feasible.
\begin{figure}[hbt]
\centering
\subfloat[AF]{
\includegraphics[width=0.33\linewidth]{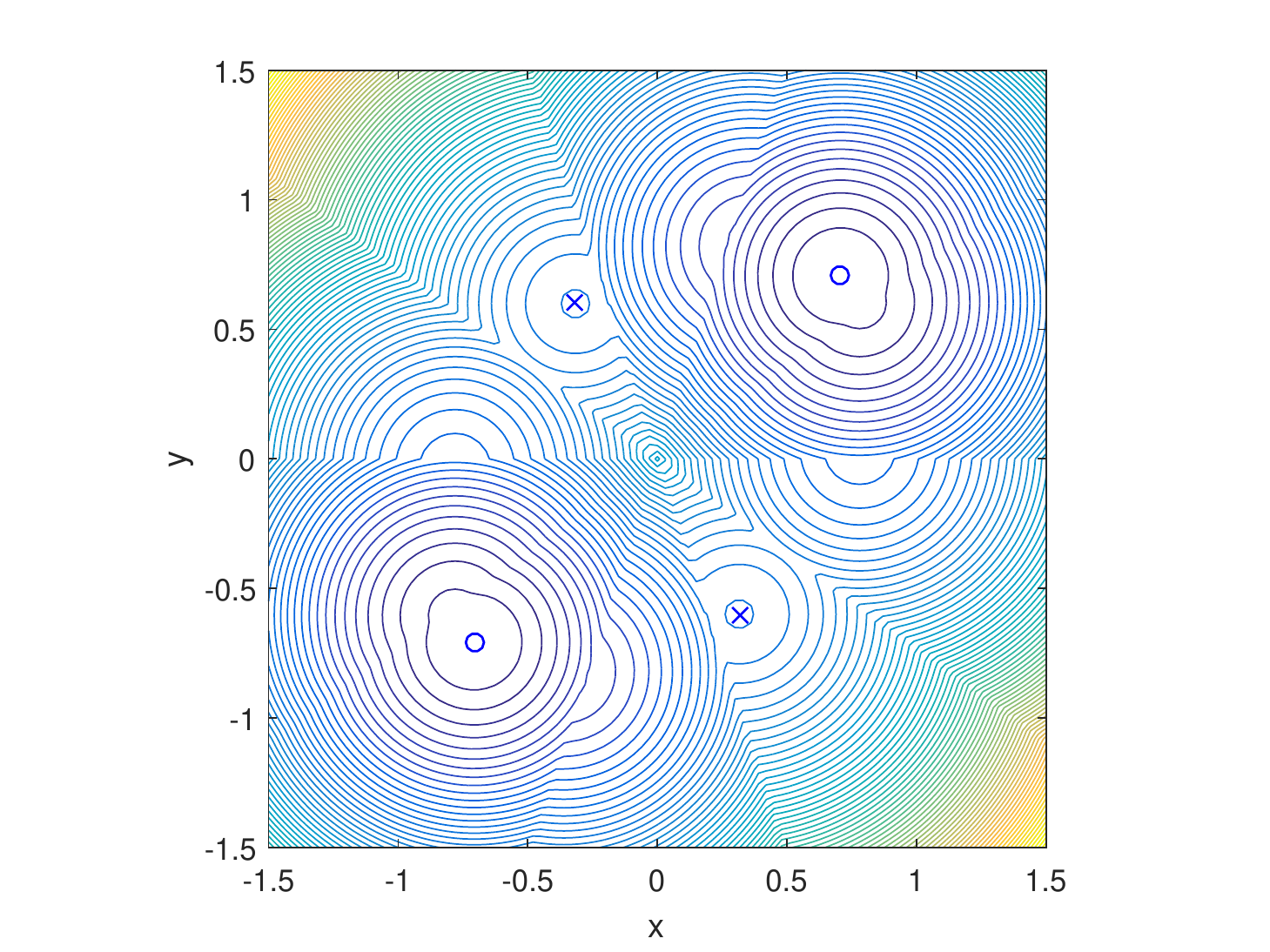}
\label{fig:af loss}
}
\subfloat[PR-SCG]{
\includegraphics [width=0.33\linewidth]{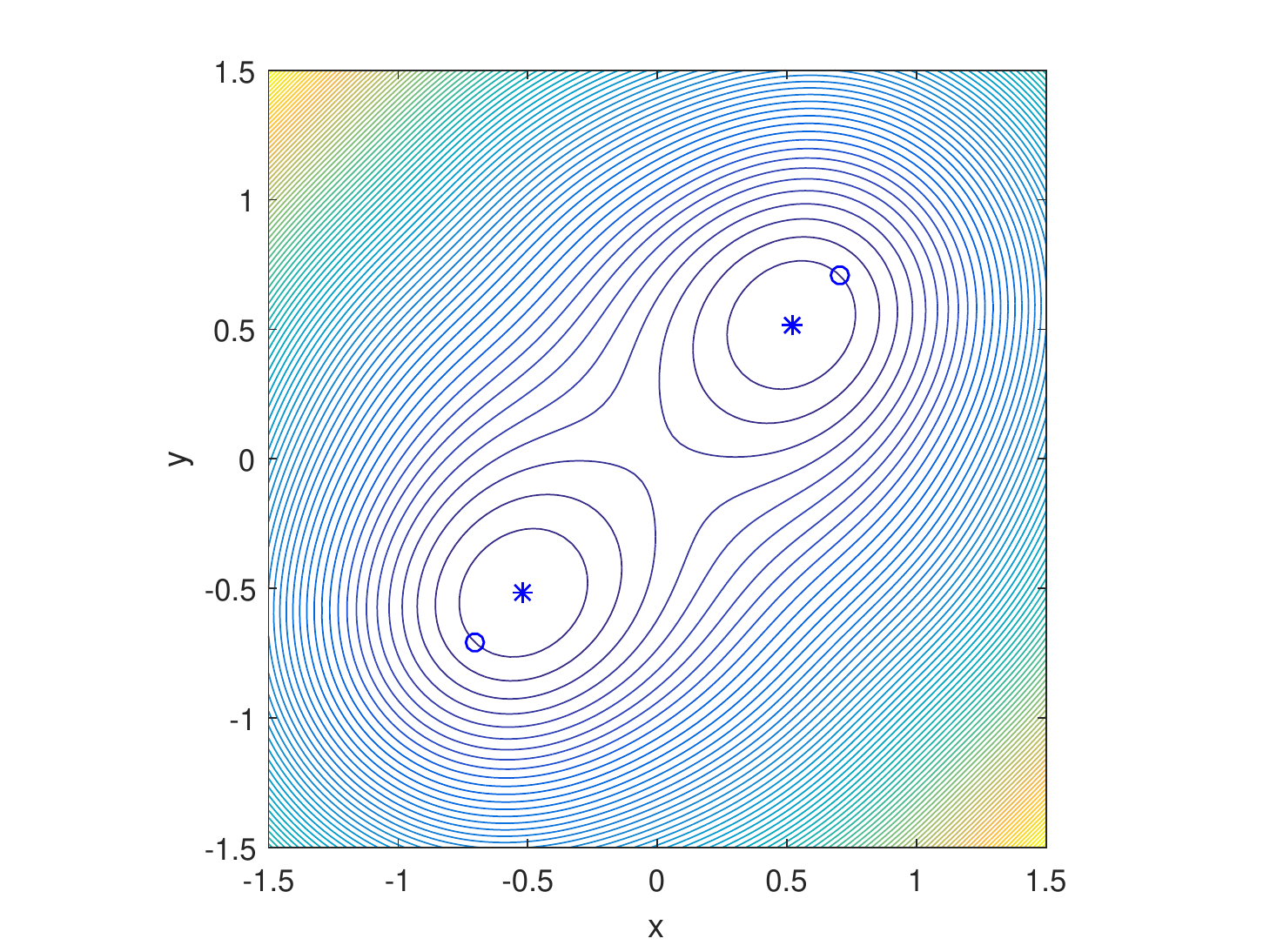}
\label{fig:scg loss}
}
\subfloat[SAF]{
\includegraphics [width=0.33\linewidth]{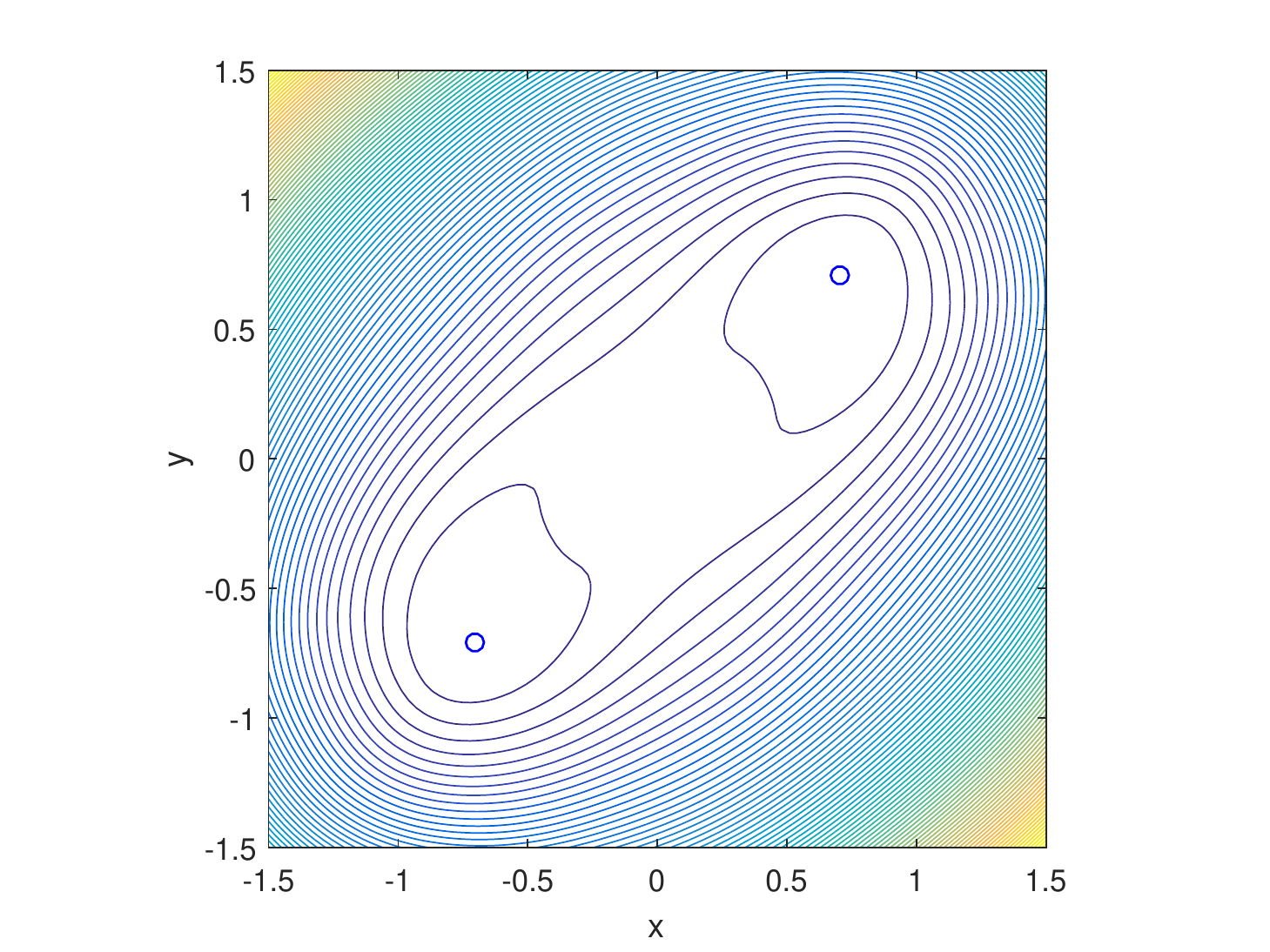}
\label{fig:saf loss}
}
\caption{Contour plot of loss functions. SAF has no local minimum while AF has another 2 spurious minima marked with crosses.}
\label{fig:intuition SAF better}
\end{figure}

Fig. \ref{fig:intuition SAF better} presents a two-dimensional real-valued PR example to demonstrate the difference between the loss functions of AF and SAF. The original solution is $\x=(\frac{\sqrt2}{2},\frac{\sqrt2}{2})'$, and the measuring vectors are $\a_i=(\sin\theta_i,\cos\theta_i)'$, where $\theta_i=\frac{i\pi}{5},i=0,\cdots,4$. It can be seen that two global minima marked with circles are surrounded by dense contours for AF and SAF as shown in Fig. \ref{fig:af loss} and  Fig. \ref{fig:saf loss}. However, there are two extra local minima marked with crosses for AF as shown in Fig. \ref{fig:af loss}. Furthermore, the contours of AF are not smooth or regular which can cause difficulties for some gradient-descending solvers, while SAF is infinitely differentiable which makes high order optimization methods applicable.

Fig. \ref{fig:scg loss} presents the contours of PR-SCG loss function with $\epsilon=0.3$ in \eqref{eq:scgloss}. The global optimal points is marked with `$*$' and obviously deviate from the true solution of PR. Therefore, an iteratively shrinking rule for $\epsilon$ is necessary for PR-SCG to ensure the convergency to the original solution. In contrast, the loss function of SAF has same global minima with \eqref{eq:afloss}.  Besides, current results only demonstrate that PR-SCG converges to some Clark stationary point, but the convergency to the global minima has not been established yet.

\subsection{Algorithm}
Since the loss function \eqref{eq:safloss} is non-convex, an elaborate initialization is needed to obtain a good initial estimate for global convergence. Because the initialization method is not the focus of this paper, we simply choose the weighted maximal correlation initialization method proposed by Wang \cite{raf2018}. This method first calculates $\eta = \sqrt{\sum_{i=1}^{m} b_i^2/m}$ as the estimate of $\norm{\x}$.
Then the direction of $\x$ is estimate by the leading eigenvector $\tilde{\z}$ of the matrix $\bm M := \sum_{i\in \mathcal{I}} \sqrt{b_i}\a_i\a_i' / \norm{\a_i}^2$, where $\mathcal{I}$ is the set of indices corresponding to the largest values of $b_i/\norm{\a_i}$.
Proposition 1 in \cite{raf2018} tells that the estimate ${\z_0} = \eta\tilde{z}/\norm{\tilde{z}}$ satisfying
\begin{equation}\label{eq:initial estimate}
  \dist{\z_0,{\x}} \leq \frac{1}{20} \norm{\x},
\end{equation}
with probability at least $1-C\exp(-c_1m)$, if $m\geq c\abs{\mathcal{I}} \geq c_2n$ for some constants $c_0,c_1,c_2,C$ and sufficiently large $n$.

We use the gradient descent method to search for the global minimizer from the above initializer; that is
\begin{equation}\label{eq:grad search}
  \z_{t+1}= \z_t - \mu_t \nabla \ell(\z_t),
\end{equation}
where $\mu_t$ is the learning rate. In the next section we show that a proper fixed $\mu_t$ suffices to ensure the global convergence under a good initialization. To obtain a faster convergence rate, we use the backtracking strategy to determine $\mu_t$ in numerical experiments. The details of SAF is presented in Algorithm \ref{alg:saf}.
\renewcommand{\algorithmicrequire}{\textbf{Input:}}
\renewcommand{\algorithmicensure}{\textbf{Output:}}
\begin{algorithm}
\caption{SAF: Smooth Amplitude Flow Method}
\begin{algorithmic}[1] %
\Require $\dkh{\a_i}_{i=1}^m,\dkh{b_i}_{i=1}^m$; maximum number of iterations $T$; step length ${\mu}$; backtracking parameters $\alpha,\beta$ and $s_{\max}$; truncation parameter $I$.
\State \textbf{Construct} $\mathcal{I}$ which includes indices corresponding to the $I$ largest entries among $\dkh{b_i}_{i=1}^m$.
\State \textbf{Initialize} $\z_0:= \sqrt{\frac1m\sum_{i=1}^{m} b_i^2}\,\tilde{\z}$, where $\tilde{\z}$ is the normalized leading eigenvector of $$\bm M:= \sum_{i\in\mathcal{I}} \sqrt{b_i}\frac{\a_i\a_i'}{\norm{\a_i}^2}$$
\For{$t = 0: T-1$}
    \State Compute gradient $\bm g_t = \nabla_{\text{SAF}}\ell(\z_t)$ according to \eqref{eq:safgrad}
    \State $s = 0$
    \While{$\ell_{\text{SAF}}\kh{\z_t-\beta^s \mu \bm g_t}> \ell_{\text{SAF}}(\z_t) -\alpha \beta^s \mu  \norm{\bm g_t}^2$ and $s < s_{\max}$   }
      \State $s=s+1$  \Comment{backtracking}
    \EndWhile
    \State $\mu_t =  \mu\beta^s$
    \State $\z_{t+1} = \z_t - \mu_t\bm g_t$
\EndFor
\Ensure $\z_T$
\end{algorithmic}
\label{alg:saf}
\end{algorithm}

\section{Theoretical guarantees for global convergence}
\label{sec:theo}
This section establishes the global convergence of Algorithm \ref{alg:saf} for the real Gaussian model. This proof can also be extended to the complex Gaussian naturally using the Wirtinger gradient. For simplicity, we write $\ell_{\text{SAF}}(\z)$ simply as $\ell(\z)$ in the following text.

The geometric convergence of SAF is characterized by the following theorem.

\begin{thm}\label{thm:grad converg}
  Consider the problem of finding arbitrary $\x\in\R^n$ from the phaseless measurements \eqref{eq:pr_problem} with real Gaussian measurement vectors. If $m\geq c_0 n $ and we adopt a fixed learning rate $\mu\leq \tilde{\mu}$, then with probability at least $1-C\exp(-c_1m)$, the SAF estimates $\z_t$ in Algorithm \ref{alg:saf} obey the geometric convergence:
  \begin{equation}\label{eq:geo conv}
    \dist{\z_t,\x} \leq \frac1{20}(1-\nu)^t \norm{\x}, \quad t=0,1,\cdots
  \end{equation}
  where $c_0,c_1,C>0$, $0<\nu<1$ and $\tilde{\mu}$ are some certain constants.
\end{thm}
Combining Theorem \ref{thm:grad converg} and the error of initialization \eqref{eq:initial estimate}, SAF can recover the original signal given $\mathcal{O}(n)$ measurements. Moreover, starting from an elaborate initial estimate, $\mathcal{O}(\log{1/\epsilon})$ iterations suffice to give a solution with RMSE error less than $\epsilon.$ Combined with the per-iteration complexity $\mathcal{O}(mn)$ we conclude that SAF solves PR in time $\mathcal{O}(mn\log\frac1\epsilon)$, which is proportional to the time required by the processor to read the entire data $\dkh{\A;\b}$.

The proof of Theorem \ref{thm:grad converg} hinges on proving the \textit{local regularity condition} $\operatorname{RC}(\mu,\lambda,c)$, i.e.
\begin{equation}\label{eq:reg cond}
  \innerprod{\nabla \ell(\z),\bm h}\geq
  \frac{\mu}{2}\norm{\nabla \ell(\z)}^2 + \frac\lambda 2 \norm{\bm \h}^2
\end{equation}
for all $\z$ such that $\norm{\h}=\norm{\z-\x}\leq \epsilon\norm{\x}$ for some constant $0<\epsilon<1.$ As shown in \cite{chen2015solving}, the ball $\dkh{\z:\norm{\z-\x}\leq \epsilon\norm{\x}}$ can be seen as a basin of attraction towards the global optimum; once the initialization lands into this neighborhood, geometric convergence can be guaranteed, i.e.,
\begin{equation}\label{eq:geo converg}
  \dist{\z+\mu\nabla\ell(\z),\x}^2\leq
  (1-\mu\lambda)\dist{\z,\x}^2.
\end{equation}
Evidently, Theorem \ref{thm:grad converg} holds asthmatically once the $\text{RC}(\mu,\lambda,\epsilon)$ is proved.

Lemma \ref{lem:nabla l leq} and Lemma \ref{lem:nablal h geq} in \ref{sec:loc convg proof} respectively demonstrate that
\begin{equation}\label{eq:nabla l leq}
  \norm{\nabla \ell(\z)}\leq \kh{1+\delta}\norm{\h},
\end{equation}
and
\begin{equation}\label{eq:nablal h geq}
  \innerprod{\nabla\ell(\z),\h}\geq (0.07-\epsilon)\norm{\h}^2,
\end{equation}
hold with probability at least $1-C\exp(-c_1m)$ given $m\geq c_0n$.
Using the above two bounds we can reach the regularity condition if $\mu$ and $\lambda$ satisfy
\begin{equation}\label{eq:mu lambda}
  0.07-\epsilon\geq \frac{\mu}{2}(1+\delta)^2+\frac\lambda2,
\end{equation}
which indicates an upper bound $\mu\leq 2\times 0.07 = 0.14$, which suggests the range of the step size.
In practice, the step size $\mu$ can be significantly larger while still ensuring the global convergence, since several bound results in our proof can be further tighten with more delicate techniques.

%
%

\section{Numerical results}
\label{sec:numr}
This section presents several numerical experiments to verify our theoretical analysis. We also compare SAF with other state-of-the-art gradient descent methods including WF, RWF, TAF and RAF in terms of the empirical rate of recovery and the convergence rate. All the experiments were conducted with Matlab 2016a on a personal laptop with {Intel} Core i7 6820HQ. Not only the real-valued Gaussian model, but also the complex-valued Gaussian model and the CDP model are tested. In following experiments, the learning rate
$\mu=4$ and $7$ for the real model and the complex model respectively, $\alpha=0.4$, $\beta = 0.2$,  truncation parameter $I=\lfloor\frac {3m}{13}\rfloor $ in the initialization stage and  $s_{\max}=2$  for the backtracking parameters.
The maximum iteration number $T$ is set as 5000. We define the normalized mean-square error $\text{NMSE}:= \operatorname{dist}^2(\z,\x)/\norm{\x}^2$ for numerical comparison.
To avoid the influence of initialization methods, all tested algorithms are seeded with the same maximal correlation methods.

For readers to reproduce the numerical tests conveniently, the Matlab codes are available at \url{https://github.com/qiluo10/smooth-amplitude-flow}.

%

\subsection{Comparison of empirical success rate}
We compare SAF with other state-of-the-art gradient methods for the noiseless Gaussian model under varying $m/n$ in terms of empirical success rate. Each success rate are calculated over 100 independent trials and a trial is declared successful if NMSE of the returned result is smaller than $10^{-5}$.
The results are depicted in Fig. \ref{fig:cmp}.
It is shown that less number of measurements suffice for SAF to recover the true signal $\x$ in comparison with AF method and its variants in real case. In the complex case, SAF is second only to PC-SCG by a narrow gap. It should be pointed out that the convergence to global optimum has not been established for PR-SCG.
Particularly, Fig. \ref{fig:cmp} illustrates that SAF achieves a high success rate of over 95\% when $m/n\geq1.8$ and a perfect recovery when $m/n\geq1.9$ in the real case. In the complex case, both SAF and PR-SCG achieve a recovery rate larger than $95\%$ when $m/n\geq2.8$ and a $100\%$ recovery when $m/n\geq3$.

It is noteworthy that SAF is the only method that achieves a $100\%$ with a sampling complexity lower than the information-theoretical limit in both the real and complex cases.
Moreover, it is enlightening that SAF simply minimizing a novel loss function can outperform other existing algorithms with complicated operations upon the gradient or loss function.

\begin{figure}[hbt]
\centering
\subfloat[Noiseless real-valued model with $\x$ and $\a_i$ independently sampled  from $\mathcal{N}(0,\bm I_n)$ ]{
\includegraphics[width=0.45\linewidth]{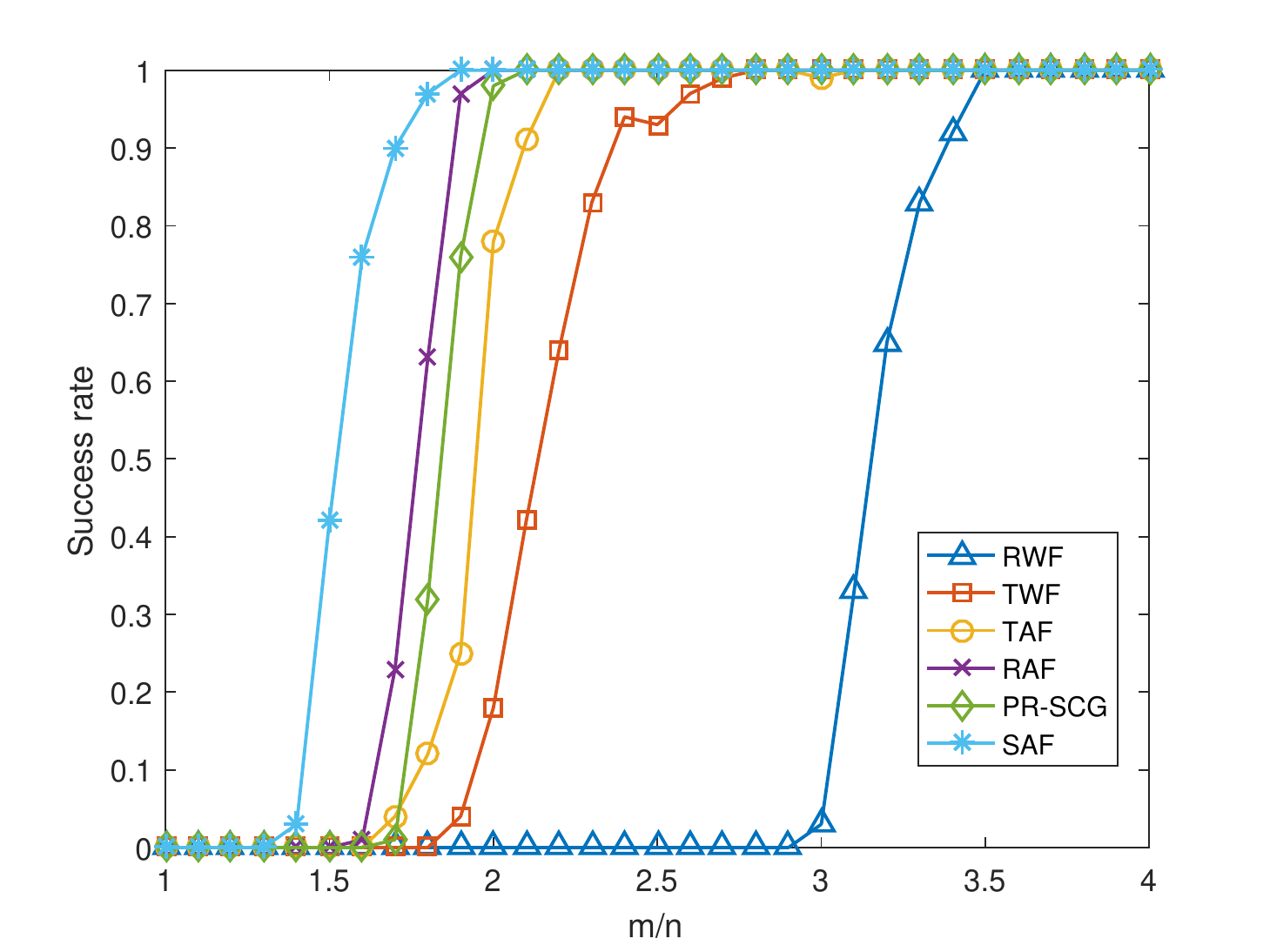}
\label{fig:real_cmp}
}
\
\subfloat[Noiseless complex-valued model with $\x$ and $\a_i$ independently sampled  from $ \mathcal{N}(0,\frac12\bm I_n)$+$j\mathcal{N}(0,\frac12\bm I_n)$.]{
\includegraphics[width=0.45\linewidth]{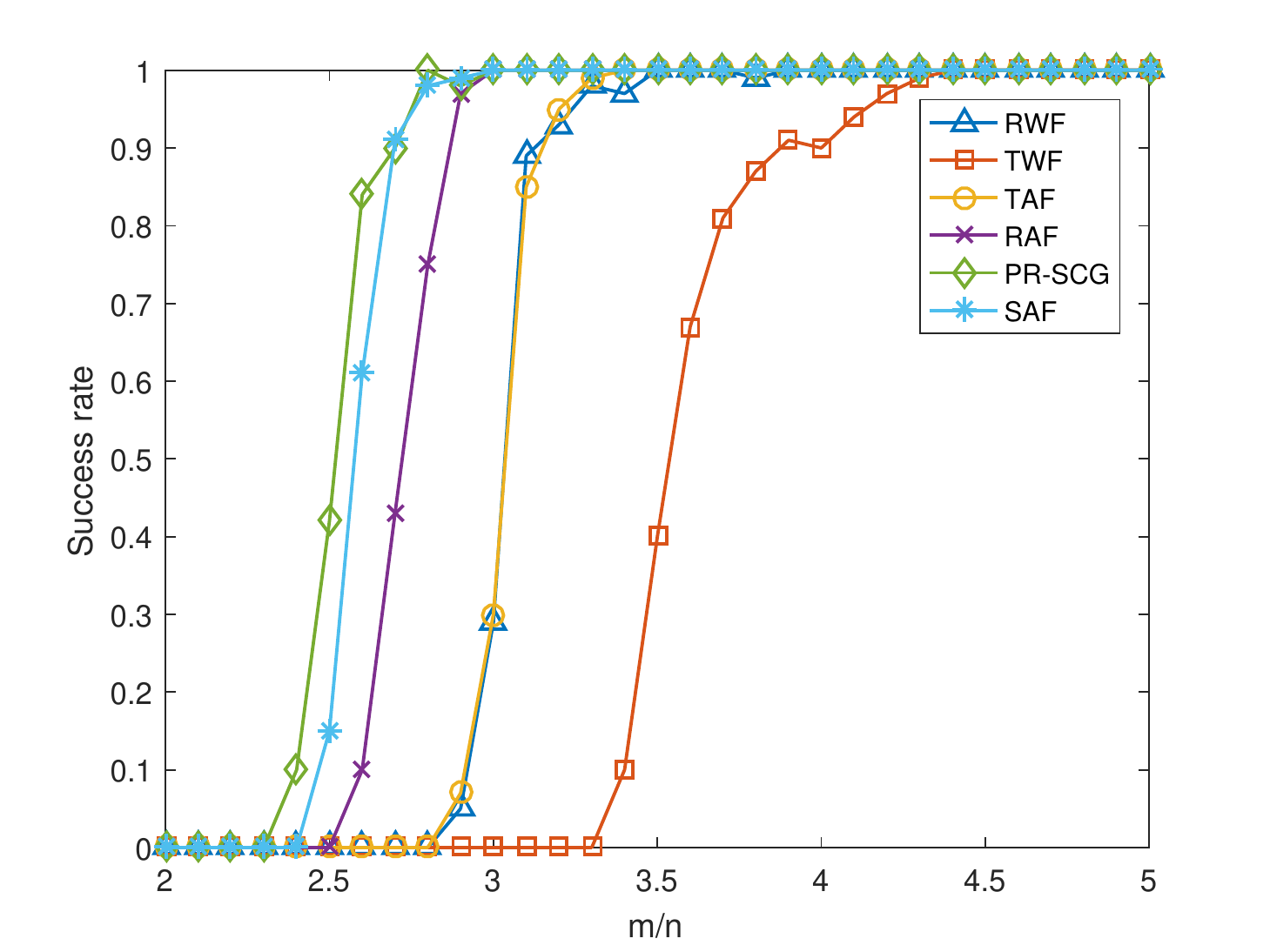}
\label{fig:cplx_cmp}
}
\caption{Empirical rate of success versus $\frac mn$ with $n=1000$. }
\label{fig:cmp}
\end{figure}

\subsection{Computational cost}

We compare the convergence rate and the time cost of SAF with other gradient descent methods, under information-theoretic limits $m=2n$ and $m=4n$  for the real and complex cases respectively. To accelerate convergence for SAF, a larger step $\mu=6$ and 10 separately for the real and complex case, while stilling achieving a $100\%$ success rate in this experiment. Table \ref{tab:convg speed} presents the number of iterations and time cost before achieving a NMSE of $10^{-14}$ for each algorithm, averaged over 100 successful trials. The optimal value is shown in bold and the second-best result is underlined in each column.
It can be observed that SAF is the second best algorithm in terms of convergence, next only to the conjugate gradient method PR-SCG. SAF ranks first and second in terms of the time cost in complex and real cases separately.

\begin{table}[hbt]
\caption{Comparison of computational costs.}
\centering
\begin{tabular}{ccccc}
\hline
\textbf{Algorithms} & \multicolumn{2}{l}{\textbf{Real Case} ($m/n=2$)}                         & \multicolumn{2}{l}{\textbf{Complex Case} ($m/n=4$)}                      \\ \hline
                   & \multicolumn{1}{l}{Iterations} & \multicolumn{1}{l}{Time (s)} & \multicolumn{1}{l}{Iterations} & \multicolumn{1}{c}{Time (s)} \\ \cline{2-5}
TWF  & - & -  & 1273.68    & 14.20  \\
RWF   & -  & - & 859.72  & {9.86} \\
TAF       & 745.51  & \textbf{0.97} & 752.74  & \underline{9.01} \\
RAF  & 1865.22 & 2.39    & 1206.10  & 13.64  \\
PR-SCG   & \textbf{138.28}    & 1.29   & \textbf{132.90}    & {9.44} \\
SAF & \underline{294.25} & \underline{1.08}& \underline{306.39} & \textbf{7.43} \\
\hline
\end{tabular}
\label{tab:convg speed}
\end{table}

\subsection{Robustness to noise}
To show the robustness of SAF against additive noise, Fig.. \ref{fig:rmse} illustrates the NMSE as a function of the signal-to-noise ratio (SNR) under different $ m/n$. The data under the Gaussian model was generated as  $b_i=(\abs{\innerprod{\a_i,\x}}^2+\eta_i)^{1/2}$, with $\eta_i$ independently {sampled} from $\mathcal{N}(0,\sigma^2)$, where $\sigma^2$ is set to achieved certain $\text{SNR}= 10\log_{10   }(\norm{\bm A\x}^2/m\sigma^2)$. For all choices of $m$, the NMSE scale decrease proportionally to the SNR, which demonstrates the stability of SAF. It is noteworthy that SAF performs robustly even under information-theoretic limit for both the real and complex Gaussian models, which has not been observed for other algorithms to the best of our knowledge.

\begin{figure}[hbt]
\centering
\subfloat[Real case]{
\includegraphics[width=0.45\linewidth]{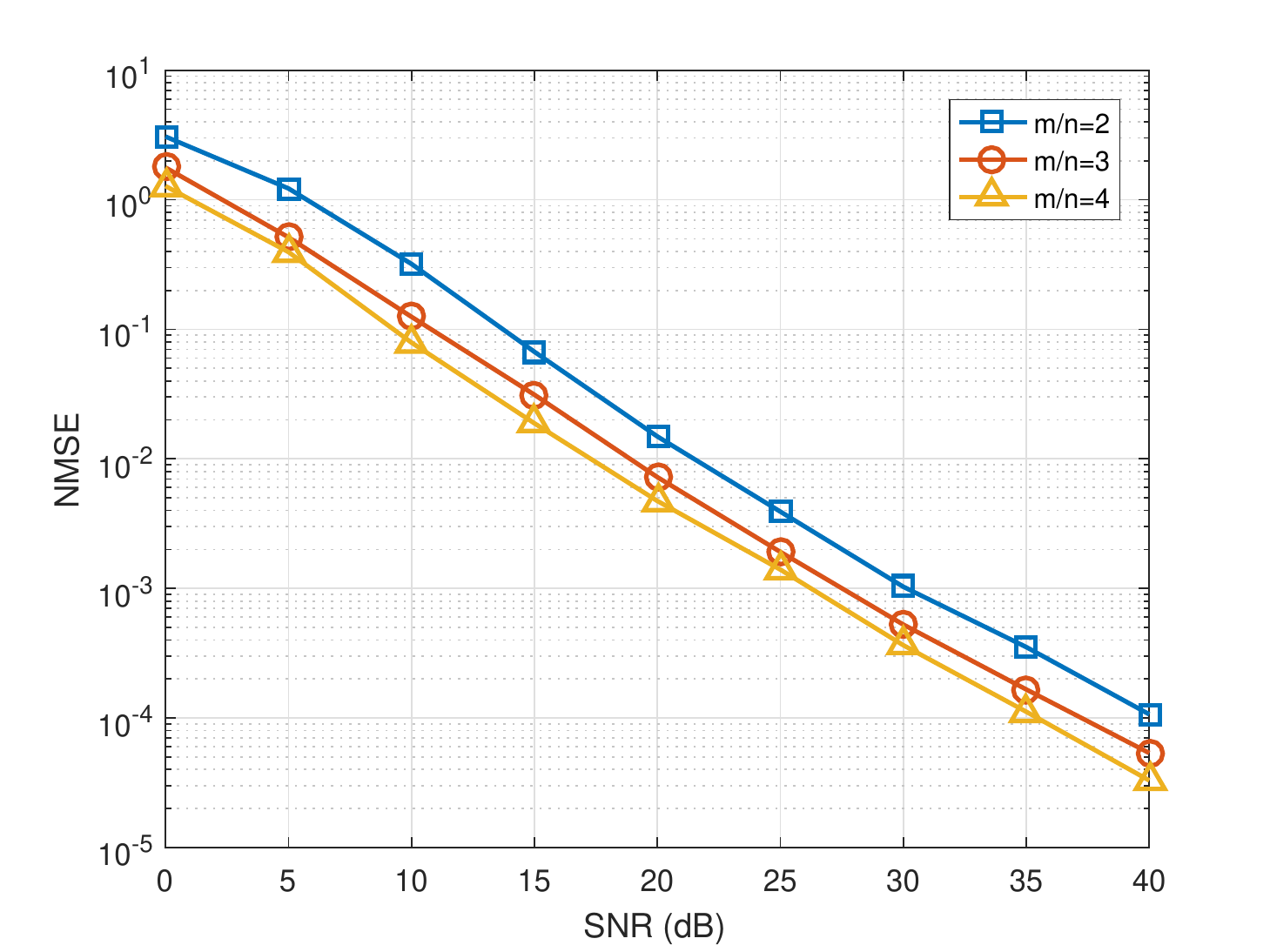}
\label{fig:real_rmse}
}
\
\subfloat[Complex case]{
\includegraphics[width=0.45\linewidth]{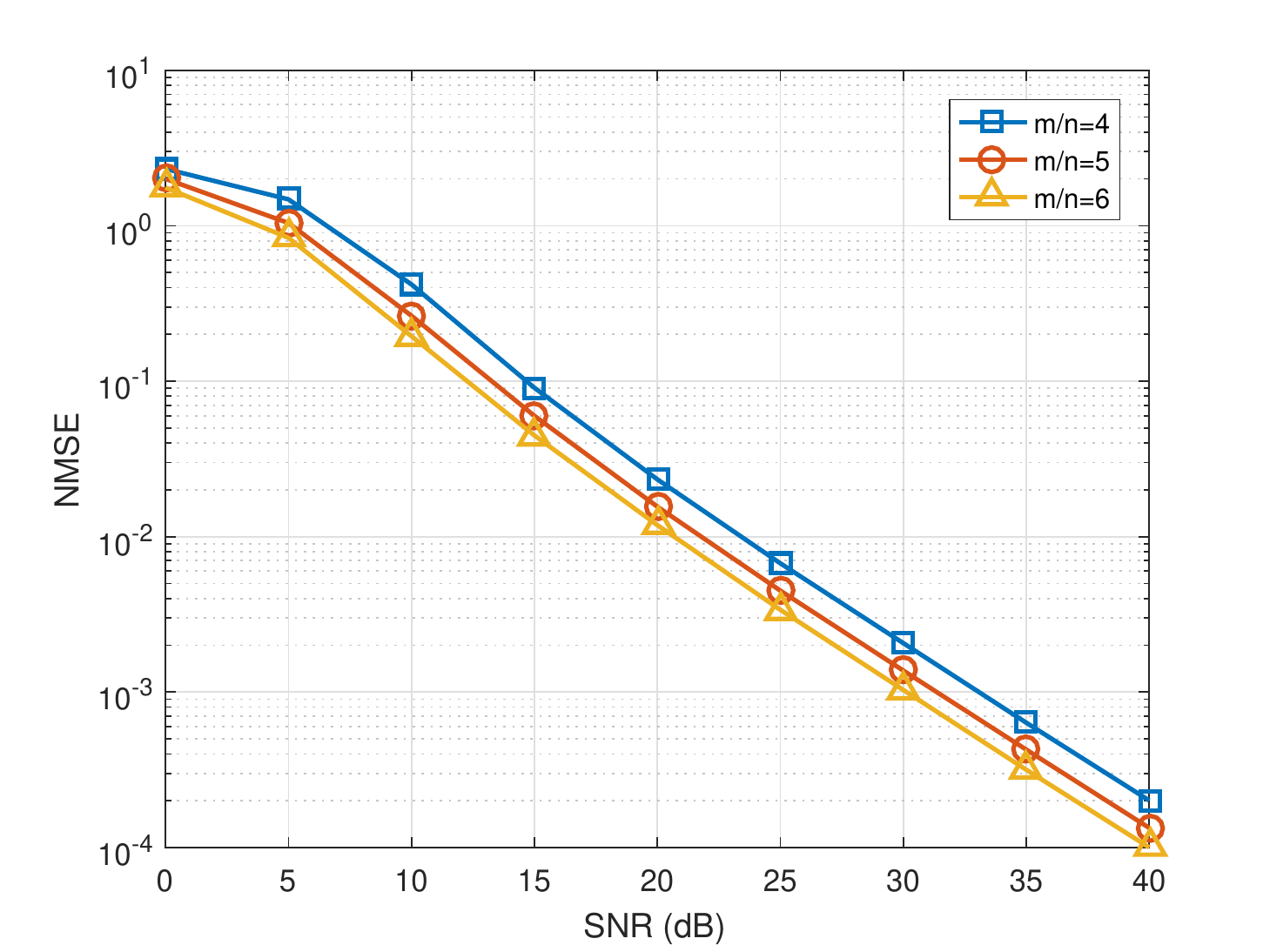}
\label{fig:cplx_rmse}
}
\caption{NMSE vs SNR for SAF under the noise Gaussian model with $n=1000$.}
\label{fig:rmse}
\end{figure}

\subsection{Image reconstructions}
To demonstrate the feasibility and scalability of SAF in phase retreival of real images, we compare SAF with other state-of-the-arts on recovering the Lena image from masked Fourier intensity measurements. This image is gray-scale that can be represented by a matrix $\bm X\in\mathbb{R}^{256\times256}$. Denoting $\x\in\mathbb{R}^n$ be a vectorization of $\bm X$, the CDP model with $K$ masks is
\begin{equation}\label{eq:cdpmodel}
  \bm b^{(k)} = |\bm{FD}^{(k)}\x|, \quad k=1,\cdots,K,
\end{equation}
where $\bm F$ represents the discrete Fourier transform matrix, and the diagonal matrix $\bm{D}^{(k)}$ is the mask, with diagonal entries sampled uniformly at random from $\dkh{1,-1,j,-j}$.
When using more than 4 masks, SAF and other exiting methods like RAF, PR-SCG with the maximum correlation initializer can recover the original image successfully. Is it possible to adopt an ever smaller $K$? Fig. \ref{fig:cdp} gives a recovered result of SAF after 200 gradient iterations in CDP model with 3 masks.
Then we compare SAF with other algorithms in terms of the recovery rate. All algorithms are seeded with the reweighted maximal correlation initialization and output results.
The comparison of success rate is shown in Table \ref{tab:cdp cmp}. It is noteworthy that all presented algorithm
SAF enjoys the highest success rate and can recover the image nearly perfectly.

\begin{figure}[hbt]
\centering
\subfloat[Ground truth]{
\includegraphics[height=0.25\linewidth]{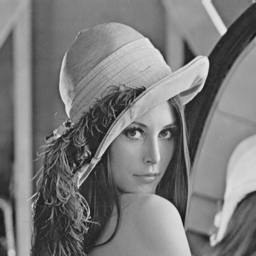}
\label{fig:original_imag}
}
\subfloat[Initialization 
]{
\includegraphics[height=0.25\linewidth]{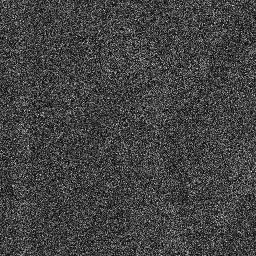}
\label{fig:ini_imag}
}
\subfloat[Recovered result]{
\includegraphics[height=0.25\linewidth]{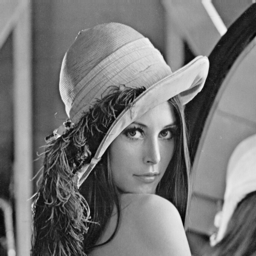}
\label{fig:saf_imag}
}
\caption{Recovered result by SAF after 200 gradient iterations from the initialization. The RMSE of initialization is 1.04 and the final RMSE is $1.04\times 10^{-3}$.}
\label{fig:cdp}
\end{figure}

\begin{table}[]
\caption{Comparisons of the success rate among algorithms on Lena image ($K$=3)}
\centering
\begin{tabular}{lcccccc}
\hline
Algorithms   & SAF  & RAF  & PR-SCG & TAF  & RWF  & TWF \\ \hline
Success Rate & \textbf{0.95} & 0.84 & 0.70   & 0.51 & 0.52 & 0   \\ \hline
\end{tabular}
\label{tab:cdp cmp}
\end{table}

\section{Conclusion}
This paper put forward SAF to solve the phase retrieval problem based on a novel smooth amplitude-based loss function. This loss function utilizes a smooth function to get rid of the non-smoothness of the original amplitude-based loss function. We prove the global geometric convergence of SAF with an elaborate initialization.
Our SAF approach is conceptually simple and can be easily implemented since it does not need extra truncating or reweighing operations upon the gradient function as many other state-of-the-art solvers do.
Substantial numerical tests were conducted and illustrated that our new loss function enjoys advantages in sampling complexity and computational efficiency.
SAF also has potential to be extended to other scenarios, e.g. recovering signals the sparsity or nonnegativity constraint.
We preliminarily analyze the reason behind the dramatic improvement brought by such a new loss function.  It will be of great interest to deeply investigate theoretical advantages of our SAF approach when measuring times is lower than the information-theoretic limit.

\section*{Acknowledgments}
This work was supported in part by National Natural Science foundation (China): 61571008. The authors would like to thank Samuel Pinilla for kindly sharing the codes and helpful discussion about PR-SCG.

\bibliographystyle{unsrt}
\bibliography{saf}

\appendix

\section{Supporting Lemmas for Section \ref{sec:theo}}\label{sec:loc convg proof}
For clarity, we give a convergence analysis for SAF with $\gamma=1,k=4$, and we conjecture that SAF adopting other reasonable setting of $\gamma$ and $k$ can also be proved with the similar routine.

For convenience, we rewrite the gradient function $\{\nabla\ell_i(z)\}_{i=1}^m$ into a unified form by a bivariate function $f$: $$\nabla \ell_i(\z) = f(\a_i'\z,b_i)\a_i,$$
where
\begin{equation}\label{eq:saf grad}
  f(x,y)= \kh{(x^4+y^4)^{1/4} -2^{1/4}y} (x^4+y^4)^{ -3/4}x^{3}.
\end{equation}
$f(\cdot,b_i)$ has the following properties:
\begin{property}\label{propty:1}
  $f(x,b_i) = f(-x,b_i)$.
\end{property}
\begin{property}\label{propty:geq0}
  For any $x\in[-b_i,b_i]$, we have $f(\pm b_i +  x ,b_i)x \geq 0 $.
\end{property}
\begin{property}\label{prop:3}
  $f(\pm b_i +  x ,b_i)x \geq 0.18 x^2 $ holds for any $x\in[-b_i/5,b_i/5]$.
\end{property}
\begin{property}\label{prop:4}
  $\abs{f(\pm b_i +  x ,b_i)/x }\leq 1$.
\end{property}
Property 1 is obvious. The details of proving other properties are put in the Appendix \ref{sec:property}. These four properties are the key intergradients in the proof of global convergence. For SAF with other setting of $\gamma>0$ and $k>2$, one can also establish these 4 properties.

The regularity condition can be proved by finding the upper bound on $\norm{\nabla \ell(\z)}$ using the property 3 in Lemma \ref{lem:nabla l leq}, and finding the lower bound on $\innerprod{\nabla\ell(\z),\h}$ based on Property 2 in Lemma \ref{lem:nablal h geq}.

\begin{lem}\label{lem:nabla l leq}
Fix $\delta>0$, given $m>c_0n$,
\begin{equation}\label{eq:temp}
  \norm{\nabla \ell(\z)}\leq \kh{1+\delta}\norm{\h}
\end{equation}
holds with probability at least $1-C\exp(-c_1m)$, where $c_0,c_1$ and $C>0$ are some universal constants.
\end{lem}
\begin{proof}
Since
\begin{equation}\label{eq:gfunc}
\begin{aligned}
   \nabla \ell_i(\z) &= f(\a_i'\x+\a_i'\h,b_i)\a_i\\
    &= \kh{f(\a_i'\x+ \a_i'\h,b_i)/(\a_i'\h) } \a_i\a_i'\h
\end{aligned}
\end{equation}
and $\abs{f(\a_i'\x+ \a_i'\h,b_i)/(\a_i'\h)}\leq 1$ , then we have
\begin{equation}\label{eq:norm l bound}
\begin{aligned}
  \norm{\nabla\ell(\z)} &= \norm{\sum_{i=1}^{m}\frac1m \kh{f(\a_i'\x+ \a_i'\h,b_i)/(\a_i'\h) } \a_i\a_i'\h }\\
  &\leq \norm{\sum_{i=1}^{m} \frac1m\kh{f(\a_i'\x+ \a_i'\h,b_i)/(\a_i'\h)} \a_i\a_i'}\norm{ \h }\\
  &\leq \norm{\sum_{i=1}^{m} \frac1m\abs{f(\a_i'\x+ \a_i'\h,b_i)/(\a_i'\h)} \a_i\a_i'} \norm{ \h }\\
  &\leq \norm{\sum_{i=1}^{m} \frac1m \a_i\a_i'} \norm{ \h }
\end{aligned}
\end{equation}
The last inequality is induced by Property 3.

By Lemma 3.1 in \cite{candes2013phaselift}, as long as  $m>c_0n$ for sufficiently large $c_0$, we have
\begin{equation}\label{eq:clsineq}
    \norm{\sum_{i=1}^{m} \frac1m \a_i\a_i'} \norm{ \h } \leq \sqrt{1+\delta} \norm{\h}\leq \kh{1+\delta}\norm{\h}
\end{equation}
with probability at least $1-C\exp(-c_1m)$.
\end{proof}

\begin{lem}\label{lem:nablal h geq}
For any sufficiently small constant $\epsilon>0$, there exist some universal constants $c_0,c_1,C$ such that, given $m\geq c_0n$,
\begin{equation}\label{eq:ellz h larger}
  \innerprod{\nabla\ell(\z),\h} \geq \tau\norm{\h}^2
\end{equation}
holds with probability at least $1-C\exp(-c_1m)$ for all $\h\in\mathbb R^n$ obeying $\norm{\h}\leq\norm{\x}/20$. Here $\tau=0.07-\epsilon$.
\end{lem}
\begin{proof}
We introduce the following events for $i=1,\cdots,m$:
$$\mathcal{A}:=\dkh{i:\abs{\a_i\h}<\frac15\abs{\a_i\x}},$$ $$\mathcal{B}:=\dkh{i:
\frac15\abs{\a_i\x}\leq\abs{\a_i\h}< \abs{\a_i\x}},$$
and $$\mathcal{C}:=\dkh{i:\abs{\a_i\h}>\abs{\a_i\x}}.$$
Using Lemma

Then we have
\begin{equation}\label{eq:tempss}
\begin{aligned}
  \innerprod{\nabla\ell(\z),\h} = \frac1m \sum_{i=1}^{m} f(\a_i'\z,b_i)\a_i'\h =
    \frac1{m} \sum_{i\in \mathcal{A}\cup\mathcal{B}\cup\mathcal{C} }f(\a_i'\z,b_i)\a_i'\h
\end{aligned}
\end{equation}
According to Property \ref{propty:geq0}, we have
\begin{equation}\label{eq:B>0}
  \frac1m\sum_{i\in\mathcal{B}}f(\a_i'\z,b_i)\a_i'\h \geq 0.
\end{equation}
Therefore
\begin{equation}\label{eq:tempssss}
\begin{aligned}
    \innerprod{\nabla\ell(\z),\h} &\geq
    \frac1m \sum_{i\in\mathcal{A}} f(\a_i'\z,b_i)\a_i'\h+\frac1m \sum_{i\in\mathcal{C}} f(\a_i'\z,b_i)\a_i'\h \\
    &\geq
    \frac1m \sum_{i\in\mathcal{A}} s\abs{\a_i'\h}^2 - \frac1{m} \sum_{i\in\mathcal{C}} \abs{f(\a_i'\z,b_i)\a_i'\h}^2 \\
    &\geq
    \frac1m \sum_{i\in\mathcal{A}} s\abs{\a_i'\h}^2 - \frac1{m} \sum_{i\in\mathcal{C}}\abs{\a_i'\h}^2\\
    &=
    \frac1m \sum_{i=1}^m s\abs{\a_i'\h}^2 -\frac1m\sum_{i\in\mathcal{A}^c}s\abs{\a_i'\h}^2   - \frac1{m} \sum_{i\in\mathcal{C}}  \abs{\a_i'\h}^2,
\end{aligned}
\end{equation}
and the last two terms has the same form and can be bounded by the Lemma \ref{lem:ah bound}. Then we have
\begin{equation}\label{eq:nabla l(z)h}
\begin{aligned}
  \innerprod{\nabla\ell(\z),\h} &\geq 0.18(\norm{\h}^2-0.25\norm{\h}^2) - 0.065\norm{\h}^2 - \epsilon\norm{\h}^2\\
  &=(0.07-\epsilon)\norm{\h}^2.
\end{aligned}
\end{equation}
\end{proof}

\begin{lem}\label{lem:ah bound}
  For any $\epsilon>0,$  given $m>c_0 n\epsilon^{-2}\log\epsilon^{-1} $, then
  \begin{equation}\label{eq:rrr}
    \frac{1}{m} \sum_{i=1}^{m} (\a_i'\h)^2 \cdot \bm 1_{\frac15|a_i'\x|\leq |\a_i'\h|}
    \leq
    ({0.25}+\epsilon)\norm{\h}^2
  \end{equation}
  \begin{equation}\label{eq:rrrs}
        \frac{1}{m} \sum_{i=1}^{m} (\a_i'\h)^2 \cdot \bm 1_{|a_i'\x|\leq |\a_i'\h|}
        \leq
        ({0.065}+\epsilon)\norm{\h}^2
  \end{equation}
  hold simultaneously with probability at least $1-C\exp(-c_1\epsilon^2m)$ for all vectors $\h\in\mathbb{R}^n$ obeying $\norm{\h}\leq \frac1{20}\norm{\x}$, where $c_0,c_1,C$ are some universal constants.
\end{lem}
\begin{proof}
We only give the detailed proof of \eqref{eq:rrr} here since the proof of \eqref{eq:rrrs} is almost the same.
We first prove the bound for any fixed vector $\h$ obeying $\norm{\h}\leq \frac1{20}\norm{\x}$ and then establish a uniform bound using an $\epsilon$-net for all vectors.

To proceed, we introduce a Lipschitz function
\begin{equation}\label{eq:chi func}
  \chi_i(t): = \left\{
    \begin{array}{ll}
      t, &\text{if } t> r^2(\a_i'\x)^2;  \\
      \frac1\delta (t-r^2(\a_i'\x)^2) + r^2(\a_i'\x)^2, &\text{if } (1-\delta)(\a_i'\x)^2\leq t\leq r^2(\a_i'\x)^2;  \\
      0, &\text{else,}
    \end{array}
  \right.
\end{equation}
for $i=1\cdots,m$, where $r>0$. The Lipschitz constant of $\chi_i(t)$ is $\frac1\delta.$
We further have
\begin{equation}\label{eq:chi ineq}
  \kh{\a_i'\h}^2 \bm 1_{r|\a_i'\x|\leq |\a_i'\h|} \leq \chi_i(|\a_i'\h|^2)
  \leq
  |\a_i'\h|^2 \bm 1_{\sqrt{1-\delta}\,r|\a_i'\x|\leq |\a_i'\h|}
\end{equation}
For convenience, we denote $\theta:= \norm{\h}/\norm{\x}$ and construct random variables $$\gamma_i:= \frac{|\a_i\h|^2}{\norm{\h}^2} \bm 1_{\sqrt{1-\delta}r|\a_i'\x|\leq |\a_i'\h|}.$$
We next compute the expectation of $\gamma_i$, via the conditional expectation,
\begin{equation}\label{eq:exp gamma}
  \E[\gamma_i] = \iint_{-\infty}^{\infty} \E\zkh{\gamma_i\big|\a_i'\x=\tau_1\norm{\x}, \a_i'\h=\tau_2\norm{\h}}q(\tau_1,\tau_2)d\tau_1 d\tau_2,
\end{equation}
where $q(\tau_1,\tau_2)$ is the probability density of two joint Gaussian random variables with correlation $\rho=\frac{\h'\x}{\norm{\h}\norm{\x}} \neq \pm 1  $. Next, we calculate $\E[\gamma_i]$ as a function of $\rho$:
\begin{equation}\label{eq:eee}
\begin{aligned}
  \E[\gamma_i] &= \iint_{-\infty}^{\infty} \tau_2^2 \cdot\bm1_{\sqrt{1-\delta} r |\tau_1|< |\tau_2|\theta} \cdot q(\tau_1,\tau_2) d\tau_1d\tau_2\\
  &=
  \frac{1}{2\pi\sqrt{1-\rho^2}} \int_{-\infty}^{\infty}\tau_2^2 \exp\kh{-\frac{\tau_2^2}{2}} \int_{-\frac{-|\tau_2|\theta}{r\sqrt{1-\delta}} }^{-\frac{|\tau_2|\theta}{r\sqrt{1-\delta}}} \exp\kh{-\frac{(\tau_1-\rho\tau_2)^2}{2(1-\rho^2)}}d\tau_1d\tau_2\\
  &=
  \frac{1}{2\pi} \int_{-\infty}^{\infty}\tau_2^2\exp\kh{-\frac{\tau_2^2}{2}}
  \int_{\frac{-\frac{|\tau_2|\theta}{r\sqrt{1-\delta} }- \rho\tau_2}{\sqrt{1-\rho^2}} }^{  \frac{\frac{|\tau_2|\theta}{r\sqrt{1-\delta} }- \rho\tau_2}{\sqrt{1-\rho^2}} } \exp\kh{-\frac{\tau^2}{2}}d\tau d\tau_2\\
  &=
  \frac{1}{2\pi} \int_{-\infty}^{\infty}\tau_2^2\exp\kh{-\frac{\tau_2^2}{2}} \cdot \sqrt{\frac{\pi}{2}} \kh{\erf \kh{\frac{\frac{|\tau_2|\theta}{r\sqrt{1-\delta}}-\rho\tau_2 }{\sqrt{2(1-\rho^2)}}}
  -\erf \kh{-\frac{\frac{|\tau_2|\theta}{r\sqrt{1-\delta}}-\rho\tau_2 }{\sqrt{2(1-\rho^2)}}}}d\tau_2 \\
  &= \frac{1}{2\pi} \int_{0}^{\infty}\tau_2^2\exp\kh{-\frac{\tau_2^2}{2}}
  \kh{\erf \kh{\frac{  \kh{\frac{\theta}{r\sqrt{1-\delta}}-\rho} \tau_2 }{\sqrt{2(1-\rho^2)}}    }
  +\erf \kh{\frac{   \kh{\frac{\theta}{r\sqrt{1-\delta}}+\rho} \tau_2 }{\sqrt{2(1-\rho^2)}}}}d\tau_2
\end{aligned}
\end{equation}
To prove \eqref{eq:rrr}, we let $r=\frac15$ and take $\delta=0.01$ and $\theta=1/20$. Then for $\rho=\pm1$, $\E[\gamma_i]=0$. As for $\rho\in(-1,1)$, we calculate it numerically using Mathematica. The result is shown in Fig. \ref{fig:egammafig} and we can see that $ \E[\gamma_i]\leq 0.25$ for $\rho\in[-1,1]$.

Moreover, the second term of \eqref{eq:eee} indicates that $\E[\gamma_i]$ is increasing with $\theta$. Therefore we have $\E[\gamma_i] \leq 0.25$ for $\theta\leq\frac1{20}$ and $\delta=0.01$, which further indicates $\E[\chi_i(|\a_i'\h|^2)]\leq0.25$ for $\theta<1/20$ and $\delta=0.01$.
\begin{figure}
  \centering
  \includegraphics[width=0.5\linewidth]{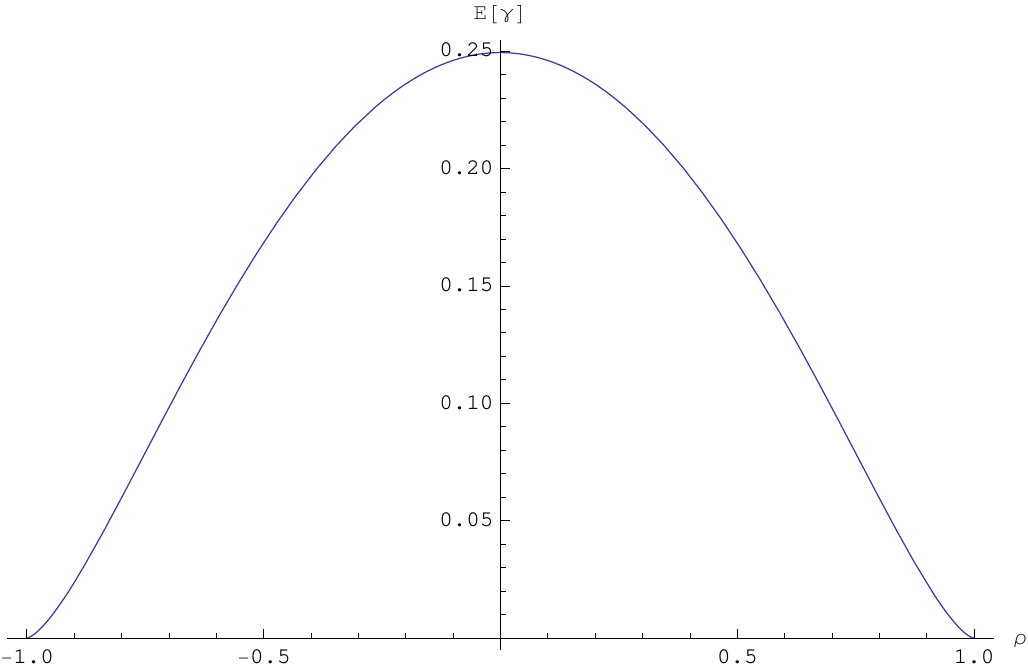}
  \caption{The figure of $\E[\gamma_i]$ for $\theta=1/20,$ with respect to $\rho$.}\label{fig:egammafig}
\end{figure}
Furthermore, $\chi_i(\abs{ \a_i'\h}^2)$ is sub-exponential. By the Bernstein-type sub-exponential tail bound \cite{eldar2012compressed},
\begin{equation}\label{eq:bernstein ineq}
  \frac{1}{m} \sum_{i=1}^{m} \frac{\chi_i(\abs{ \a_i'\h}^2)}{\norm{\bm h}^2} \leq  0.25+\epsilon
\end{equation}
holds with probability at least $1-\exp(-c_1m\epsilon^2)$ for some universal constant $c_1$ if $\norm{\h}\leq \frac{1}{20}\norm{\x}$.

Thus we have proved that \eqref{eq:rrr} holds for a fixed $h$. Now we show that this claim about a fixed $\h$ can be extended to all $\h$. We prove this claim on the sphere $\mathcal S:=\dkh{\h:\norm{\h}=\norm{\x}/20},$ and then explain this remains true inside this sphere.

Let $\epsilon'= \frac{1}{20 }\epsilon\|\x\|$, and we construct a $\epsilon'$-net $\mathcal{N}_{\epsilon'}$ covering the sphere $\mathcal{S}$. The cardinality $\abs{\mathcal{N}_{\epsilon'}}\leq \kh{1+\frac{2}{\epsilon}}^n$. For any vector $\h \in \mathcal{S}$, there exists $\h_0\in \mathcal{N}_{\epsilon'}$ satisfying $\norm{\h-\h_0}\leq\epsilon'=\epsilon\norm{\h}$. For all points on the $\mathcal{N}_{\epsilon'}$, we have
\begin{equation}\label{eq:union  bound}
  \frac{1}{m} \sum_{i=1}^{m} \kh{\abs{\a_i'\h_0}^2  } \leq (0.31+\epsilon)\norm{\h_0}^2,\quad\forall \h_0\in\mathcal{N}_{\epsilon'}
\end{equation}
with probability at least $1-\kh{1+\frac{2}{\epsilon}}^2\exp(-c_1m\epsilon^2)$. By Lemma 1 and Lemma 2 in \cite{chen2015solving},
\begin{equation}\label{eq:chemlem1}
  \sum_{i=1}^m\frac{1}{m}\kh{\abs{\a_i'\h}^2-\abs{\a_i'\h_0}^2 }
  \leq c_2\norm{\h'\h-\h_0'\h}_F \leq 3 \norm{\h-\h_0}
\end{equation}
holds with probability at least $1-C\exp(-c_1m)$ as long as $m>c_0n$, with some universal constants $C,c_0,c_1,c_2>0$.
Then we have
\begin{equation}\label{eq:2eee}
  \begin{aligned}
    &\abs{\frac{1}{m}\sum_{i=1}^{m}\chi_i\kh{\abs{\a_i'\h}^2} - \frac{1}{m}\sum_{i=1}^{m}\chi_i\kh{\abs{\a_i'\h_0}^2}}\\
  \leq
    &\sum_{i=1}^{m}\frac1m \abs{ \chi_i\kh{\abs{\a_i'\h}^2}-\chi_i\kh{\abs{\a_i'\h_0}^2}}\\
  \leq
    &\frac{1}{\delta}
    \sum_{i=1}^m\frac{1}{m}\norm{\abs{\a_i'\h}^2-\abs{\a_i'\h_0}^2 }\qquad \text{($\chi_i(t)$ is $\frac1\delta$-Lipschitz)}\\
  \leq
    &c_2 \frac1\delta \norm{\h'\h-\h_0\h}_F
    \leq
    \frac3\delta \norm{\h-\h_0}\cdot\norm{\h}\leq \frac{3c_2\epsilon}{\delta}\norm{\h}^2
  \end{aligned}
\end{equation}
On the event that both \eqref{eq:union  bound} and \eqref{eq:2eee} hold, we have
\begin{equation}\label{eq:eq49}
  \frac1m \sum_{i=1}^{m}\chi_i(\abs{\a_i'\h}^2)\leq (0.25+\epsilon+3c_2\epsilon/\delta)\norm{\h}^2,
\end{equation}
for all $\norm{\h}$ with $\norm{\h}=\frac{1}{20}\norm{\x}$.

As for the situation when $\norm{\h}< \frac1{20}\norm{\x}$, find $\h=\frac{1}{w}\h'\in \mathcal{S}$ and $w\in(0,1)$. One can easily verify that
\begin{equation}\label{eq:chiii}
\begin{aligned}
   \chi_i(\abs{\a_i'\h}^2) &= \chi_i(\abs{\a_i'w\h}^2)\leq w^2 \chi_i(\abs{\a_i'\h}^2)\\
  &\leq
   (0.25+\epsilon+3c_2\epsilon/\delta)\norm{\h'}^2,
\end{aligned}
\end{equation}
on the same event that \eqref{eq:eq49} holds.

So far we have proved the \eqref{eq:rrr}. One can find that  $\E[\gamma_i]\leq 0.065$ when $r$ takes 1 and reproduce the proof of \eqref{eq:rrr} to prove \eqref{eq:rrrs}.
\end{proof}

\section{The proof of properties of SAF loss function}\label{sec:property}
Property \ref{propty:1} holds because
\begin{equation}\label{eq:prop1 ill}
  f(-x,y) =  \kh{((-x)^4+y^4)^{1/4} -2^{1/4}y} ((-x)^4+y^4)^{ -3/4}(-x)^{3} = -f(x,y).
\end{equation}

As for Property 2, we have
\begin{equation}\label{eq:192}
\begin{aligned}
  f(b_i+x,b_i) x &= \kh{((b_i+x)^4+b_i^4)^{1/4} -2^{1/4}b_i} ((b_i+x)^4+b_i^4)^{ -3/4}(b_i+x)^{3}x.
\end{aligned}
\end{equation}
$((b_i+x)^4+b_i^4)^{ -3/4}$ is always non-negative. $(b_i+x)^3\geq 0$ when $x\geq-b_i$. And it is easy to check that
$$\kh{((b_i+x)^4+b_i^4)^{1/4} -2^{1/4}b_i}x\geq0$$ for any $x$. Therefore the $f(b_i+x)x\geq 0$ holds for all $x\geq-b_i$. Because of Property 1, we have $$f(-b_i+x,b_i)x = f(b_i-x,b_i)x\geq0$$
when $x\leq b_i$. Hence Property 2 is proved.

It can be difficult to prove the Property 3 analytically. Here we demonstrate this property by directly plotting the graph of $g(x)= f( 1 +  x,1)x - 0.18x^2$ in Fig. \ref{fig:property2}.
\begin{figure}
  \centering
  \includegraphics[width=0.5\linewidth]{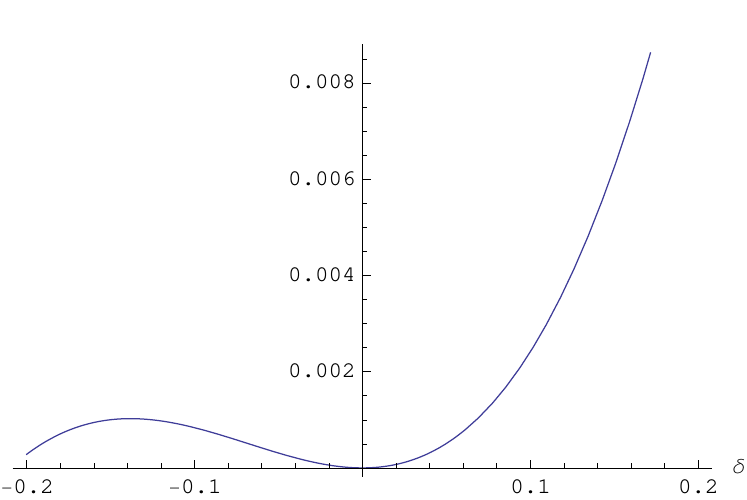}
  \caption{Plot of $g(x)$}\label{fig:property2}
\end{figure}

Now we focus on Property 4.
Because of Property 1, we need only prove the case $|f(b_i+\delta,b_i)/\delta|\leq 1$. Then Property 4 holds since
\begin{equation}\label{eq:tempppp}
 \begin{aligned}
  |f(b_i+\delta,b_i)|&= \abs{((b_i+\delta)^4+b_i^4)^{1/4} -2^{1/4}b_i} \frac{\abs{b_i+\delta}^3}{((b_i+\delta)^4+b_i^4)^{ 3/4}}\\
  &= \abs{((b_i+\delta)^4+b_i^4)^{1/4} -(b_i^4+b_i^4)^{1/4}} \frac{\abs{b_i+\delta}^3}{((b_i+\delta)^4+b_i^4)^{ 3/4}}\\
  &\leq \abs{\abs{b_i+\delta} - b_i} \frac{\abs{b_i+\delta}^3}{((b_i+\delta)^4+b_i^4)^{ 3/4}}\\
  &\leq \abs{\delta}.
  \end{aligned}
\end{equation}

\end{document}